%% file: main.tex
\documentclass[letterpaper, 10 pt]{IEEEtran}
\IEEEoverridecommandlockouts

\input{packages}

\input{notation}

\begin{document}

\title{Supervised Learning for Stochastic Optimal Control}

\author{Vince Kurtz and Joel W. Burdick

\thanks{The authors are with the Department of Civil and Mechanical Engineering, California Institute of Technology, \texttt{\{vkurtz,jwb\}@caltech.edu}}}

\maketitle
\thispagestyle{empty}

\begin{abstract}
    Supervised machine learning is powerful. In recent years, it has enabled massive breakthroughs in computer vision and natural language processing. But leveraging these advances for optimal control has proved difficult. Data is a key limiting factor. Without access to the optimal policy, value function, or demonstrations, how can we fit a policy? In this paper, we show how to automatically generate supervised learning data for a class of continuous-time nonlinear stochastic optimal control problems. In particular, applying the Feynman-Kac theorem to a linear reparameterization of the Hamilton-Jacobi-Bellman PDE allows us to sample the value function by simulating a stochastic process. Hardware accelerators like GPUs could rapidly generate a large amount of this training data. With this data in hand, stochastic optimal control becomes supervised learning.
\end{abstract}

\input{01_intro}
\input{02_related}
\input{03_problem}
\input{04_generating}
\input{05_supervised}
\input{06_examples}
\input{08_conclusion}

\bibliographystyle{IEEEtran}
\bibliography{references}

\end{document}

%% file: packages.tex
\usepackage{cite}
\usepackage{amsmath,amssymb,amsfonts}
\usepackage{bm}
\usepackage{graphicx}
\usepackage{textcomp}
\usepackage{xcolor,color}
\usepackage{amsthm}
\usepackage[ruled,vlined]{algorithm2e}
\usepackage[top=0.833in, left=0.667in, right=0.667in, bottom=0.597in, includefoot]{geometry}
\usepackage{hyperref}
\usepackage{tikz}
\usetikzlibrary{positioning,automata}

%% file: notation.tex
\newtheorem{theorem}{Theorem}

\newtheorem{remark}{Remark}

\newtheorem{assumption}{Assumption}

%% file: 01_intro.tex
\section{Introduction}\label{sec:intro}

Supervised learning, which is essentially regression with large datasets, has revolutionized computer vision and natural language processing. Tasks that were recently considered nearly impossible---like identifying objects in a scene \cite{fu2019dual}, generating text \cite{vaswani2017attention}, and producing photo-realistic images \cite{song2020score}---are now essentially solved. These advances are being rapidly integrated in consumer products, and have launched a multi-billion dollar ``Artificial Intelligence'' industry \cite{rudnitsky2023chat}. 

Optimal control and robotics have not seen a similar explosion in progress, despite broad interest in leveraging machine learning \cite{bertsekas2019reinforcement, ibarz2021train, peters2019data}. Reinforcement learning (RL), the most well-known attempt to do so, has produced impressive demonstrations \cite{chen2023visual, hoeller2024anymal}, but remains brittle and difficult to use effectively. State-of-the-art methods like Proximal Policy Optimization (PPO) \cite{schulman2017proximal} are sensitive not only to major elements like reward specification, but also to minor implementation details and even the random seed used for training \cite{engstrom2019implementation, andrychowicz2020matters}. Blog posts with titles like ``The 37 Implementation Details of PPO" are currently essential reading for would-be RL practitioners \cite{huang202237}. 

What makes machine learning for control so difficult? One reason is data. Computer vision and natural language processing have massive image and text datasets from the internet. For control, it is less obvious where to get the data. One option is expert (e.g., human) demonstrations. While this approach is unlikely to generate internet-scale training data, diffusion policy methods show significant promise along this path, particularly for dexterous manipulation \cite{chi2023diffusion}.

For tasks that are difficult to demonstrate, options are more limited. On the one hand, simulation data is readily available. GPU-based simulators like Isaac Sim \cite{makoviychuk2021isaac}, Brax \cite{freeman2021brax}, and MuJoCo XLA \cite{mjx} enable massively parallel simulation of complex systems. Indeed, these high-throughput simulators underpin much recent progress in RL \cite{rudin2022learning}. But supervised learning cannot be applied to simulation data directly. RL is not supervised learning, and as such does not admit even the modest local convergence guarantees that accompany supervised learning \cite{bottou2018optimization}. This leads to brittleness in practice, as the coupled process of data collection and optimization can diverge. In contrast, supervised learning tends to be more robust to implementation details and hyperparameter choices. 

\begin{figure}
    \centering
    \includegraphics[width=0.75\linewidth]{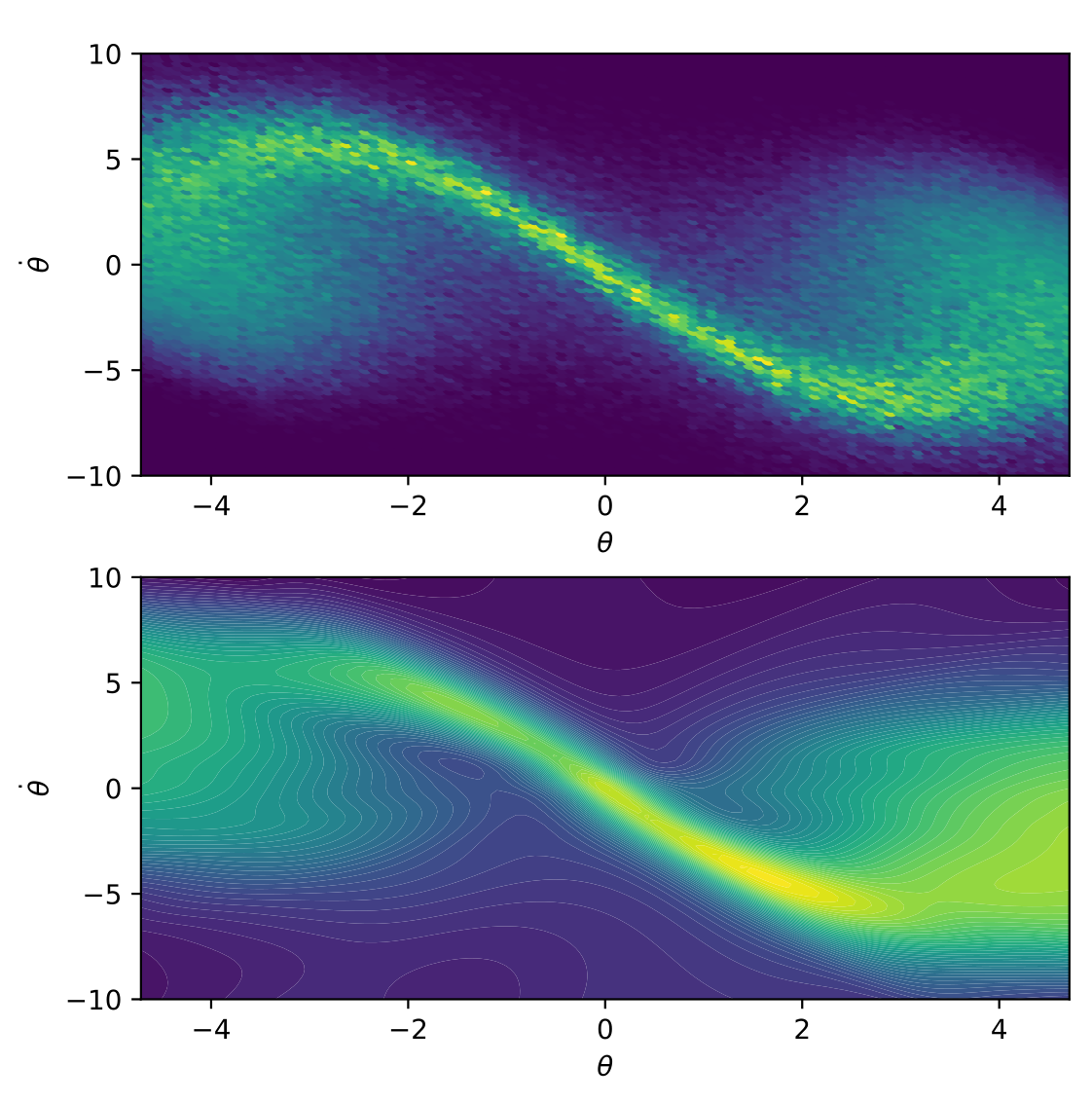}
    \caption{Training data (top) and neural network fit (bottom) for an inverted pendulum swing up task. Light yellow indicates a high desirability score---a transformation of the value function. Training data is generated using only simulations with random inputs, but exhibits the characteristic ``swirl'' of the value function for this nonlinear system. }
    \label{fig:front_page}
\end{figure}

In this paper, we cast stochastic optimal control as supervised learning. In particular, we show how to generate training data for a class of continuous-time nonlinear stochastic optimal control problems. This training data is produced offline using only simulations with random inputs: no expert policy or demonstrations are needed. Once we have this training data, policy optimization becomes a regression problem. 

Our basic idea starts with the Hamilton-Jacobi-Bellman (HJB) equation, a nonlinear partial differential equation (PDE) defining the value function. We then reformulate the HJB equation as a linear PDE using the desirability transformation \cite{kappen2005linear, todorov2009efficient, leong2016linearly}. The parabolic structure of this desirability PDE allows us to sample the desirability score via Monte-Carlo simulation of a stochastic process, via the Feynman-Kac theorem. These desirability samples are the training data. We then use these samples to fit a desirability model, from which we can obtain the value function and therefore the optimal policy. 

The remainder of the paper is organized as follows. Section~\ref{sec:related} reviews related work and Section~\ref{sec:problem} presents a formal problem statement. We detail how to obtain training data in Section~\ref{sec:training_data} and present a policy optimization algorithm that uses this data in Section~\ref{sec:supervised}. Section~\ref{sec:examples} provides two small examples, and we conclude with Section~\ref{sec:conclusion}.

%% file: 02_related.tex
\section{Related Work}\label{sec:related}

Most learning-based optimal control methods assume discrete-time dynamics \cite{schulman2017proximal, lillicrap2015continuous}. These and similar RL techniques have produced incredibly impressive results on difficult optimal control problems like robot locomotion \cite{rudin2022learning, hoeller2024anymal} and dexterous manipulation \cite{chen2023visual}, but remain brittle to hyperparameter tuning and difficult to apply effectively to new systems \cite{ibarz2021train, engstrom2019implementation, andrychowicz2020matters}.

A less popular but nonetheless promising set of alternatives focus on continuous time, where structure in the system dynamics and optimality conditions becomes more clear. Prominent examples include continuous fitted value iteration \cite{lutter2021value} and robust variants \cite{lutter2021robust}. Other work in this line attempts to embed the structure of the HJB equation in the value approximation of discrete-time methods like PPO \cite{mukherjee2023bridging}.

Solving the HJB equation is, unsurprisingly, a unifying theme \cite{shilova2023revisiting}. If we can solve the HJB PDE for the value function, we immediately obtain an optimal controller. Value iteration is the classical method for doing so \cite{mitchell2005time}, but while convergence guarantees for discrete action spaces \cite{tsitsiklis1996analysis} or linear function approximators \cite{munos2008finite} are available, no such guarantees are available for general-purpose approximators like neural networks. In practice, value iteration can not only converge to a poor local minimum, but even diverge entirely. 

Alternative attempts to solve the HJB PDE include the application of physics-informed neural networks, finite element methods, and other numerical techniques \cite{shilova2023revisiting, lutter2020hjb, schiassi2022bellman, mitchell2005toolbox}. These aim to solve HJB directly, but the PDE's nonlinearity presents a serious challenge, especially in high-dimensions.

In the stochastic case, the nonlinear HJB can be linearized via a log transformation \cite{kappen2005linear,todorov2009efficient}. The resulting PDE is still high-dimensional and second order, but its linearity enables solutions via a new set of numerical techniques. These include sum-of-squares optimization \cite{leong2016linearly}, tensor methods \cite{horowitz2014linear}, and path integral techniques \cite{theodorou2010generalized}. Path integral methods are a particular source of inspiration for this work, as they also draw on the Feynman-Kac theorem to derive an optimal controller. They tend to focus on online optimization \cite{williams2016aggressive}, unsupervised policy learning \cite{theodorou2010generalized}, or RL with motion primitives \cite{stulp2012reinforcement}, however, rather than offline supervised learning.

%% file: 03_problem.tex
\section{Problem Statement}\label{sec:problem}

In this paper, we consider continuous-time stochastic systems with control-affine dynamics of the form
\begin{equation}\label{eq:dyanmics}
    dx_t = ( f(x_t) + G(x_t)u_t )dt + B(x_t)d\omega_t,
\end{equation}
where $x_t \in \Omega \subset \mathbb{R}^n$ is the state at time $t$, and evolves on the compact domain $\Omega$. $u_t \in \mathbb{R}^m$ is the control input, and $\omega_t \in \mathbb{R}^l$ is Brownian noise with covariance $\Sigma_\varepsilon$, i.e., $\omega_t - \omega_s \sim \mathcal{N}(0, \Sigma_\varepsilon(t-s))$. We assume that $f$, $G$, and $B$ are Lipschitz. 

Our goal is to find a control policy, $u(t)$ for $t\in[0,T]$, that minimizes the cost
\begin{equation}\label{eq:cost}
    J(x_{0:T}, u_{0:T}) = \mathbb{E}_{\omega_t}\left[\phi(x_T) + \int_{0}^T \ell(x_t) + \frac{1}{2}u_t^TRu_t ~dt \right]
\end{equation}
over a fixed time horizon $T$, where $\ell : \Omega \to \mathbb{R}_+$ and $\phi : \Omega \to \mathbb{R}_+$ are Lipschitz running and terminal costs, and $R \in \mathbb{S}_{++}^m$ is a symmetric positive definite control cost matrix. 

We make the following assumption about the noise $\omega_t$:
\begin{assumption}\label{assume:noise}
    There exists a $\lambda \in \mathbb{R}_+$ such that for all $x \in \Omega$,
    \begin{equation}\label{eq:noise_assumption}
        \lambda G(x)R^{-1}G(x)^T = B(x)\Sigma_\varepsilon B(x)^T \triangleq \Sigma(x).
    \end{equation}
\end{assumption}
This assumption, which is common in the path integral control literature \cite{theodorou2010generalized, williams2016aggressive, stulp2012reinforcement}, states that the system must have sufficient control authority along state dimensions that are impacted by noise. It also ensures that we can scale $B(x)$ by injecting noise in $u$. This renders $\lambda$ a free hyperparameter.

%% file: 04_generating.tex
\section{Generating Training Data}\label{sec:training_data}

In this section, we show how to cast the optimal control problem (\ref{eq:cost}) as a supervised learning (regression) problem. Once we have a regression problem, we can apply any number of powerful machine learning tools. 

We begin with the value function
\begin{equation}\label{eq:value}
    V(x_t, t) =  \min_{u_{t:T}} J(x_{t:T}, u_{t:T}),
\end{equation}
or optimal cost-to-go. For our system \eqref{eq:dyanmics}, the optimal policy 
\begin{equation}\label{eq:value_policy}
    u = -R^{-1} G(x)^T \nabla_x V(x, t)
\end{equation}
is easily found from the value function \eqref{eq:value}, if $V$ is available in closed form. This value function is the viscosity solution of the HJB equation \cite{yong1997relations}
\begin{multline}\label{eq:hjb}
    -\partial_t V = \ell + f^T\nabla_x V - \frac{1}{2}\nabla_x V^T G R^{-1} G^T \nabla_x V \\
    + \frac{1}{2}Tr\left( (\nabla_{xx}V) B \Sigma_\varepsilon B^T\right),
\end{multline}
with boundary condition
\begin{equation}\label{eq:hjb_boundry}
    V(x_T, t) = \phi(x_T).
\end{equation}
This second-order nonlinear PDE is notoriously difficult to solve, even numerically \cite{shilova2023revisiting}. 

Fortunately, under Assumption~\ref{assume:noise}, it is possible to transform \eqref{eq:hjb} into a linear PDE. We do so by way of the \textit{desirability} $\Psi$, which is related to the value function as
\begin{equation}\label{eq:desirability_defn}
    V = - \lambda \log \Psi.
\end{equation}
Intuitively, $\Psi$ compresses $V$ from $[0, \infty)$ to $[0, 1]$. 

\begin{remark}
    The fact that $\Psi$ takes values in $[0, 1]$ makes it an appealing target for supervised learning with neural networks, where normalization is important. Standard value function approximations are not normalized in this way, and may suffer as a result: recent work suggests that using classification---which involves outputs in $[0, 1]$---outperforms standard regression techniques for RL \cite{farebrother2024stop}.
\end{remark}

The desirability PDE obtained from \eqref{eq:hjb} and \eqref{eq:desirability_defn} is given by
\begin{equation}\label{eq:desirability_pde}
    - \partial_t \Psi = - \frac{1}{\lambda} \ell \Psi + f^T \nabla_x \Psi + \frac{1}{2}Tr\left((\nabla_{xx}\Psi)\Sigma\right),
\end{equation}
with boundary condition
\begin{equation}\label{eq:desirability_boundry}
    \Psi(x_T, t) = \exp\left(-\frac{1}{\lambda} \phi(x_T)\right) \triangleq \psi(x_T).
\end{equation}
This new PDE is linear, but a closed-form solution is still not available in general. However, the particular structure of \eqref{eq:desirability_pde} lends itself to a useful set of numerical techniques. In particular, \eqref{eq:desirability_pde} is a special case of the Feynman-Kac PDE
\begin{multline}\label{eq:fk_pde}
    \partial_t \nu(x, t) + c(x, t)\nu(x, t) + b(x, t)^T \nabla_x \nu(x, t)  \\ 
    + \frac{1}{2}Tr\left((\nabla_{xx}\nu) \sigma(x) \sigma^T(x)\right) + h(x, t) = 0,
\end{multline}
with boundary condition
\begin{equation}\label{eq:fk_boundary}
    \nu(x, T) = \varphi(x),
\end{equation}
where $\nu : \mathbb{R}^n \times [0, T] \to \mathbb{R}$. Other special cases include the Schr{\"o}dinger equation in quantum mechanics \cite{glimm2012quantum} and the Black-Scholes-Merton equation in mathematical finance \cite{brandimarte2013numerical}. 

The Feynman-Kac theorem relates the solution $\nu(x, t)$ of this PDE to the expected value of a stochastic process:
\begin{theorem}[Feynman-Kac]\label{thm:fk}
    Let $X_t$ be the solution to the stochastic differential equation
    \begin{equation}\label{eq:fk_sde}
    \begin{gathered}
        dX_s = b(X_s, s)ds + \sigma(X_s, s)dW_s, \\
        X_0 = x,
    \end{gathered}
    \end{equation}
    where $x \in \mathbb{R}^n$, $s \in [0, T]$, and $W$ is standard Brownian noise. Then the viscosity solution of (\ref{eq:fk_pde},~\ref{eq:fk_boundary}) is given by 
    \begin{multline}\label{eq:fk_expectation}
        \nu(x, t) = \mathbb{E}\Big[ \int_{t}^T h(X_s, s) e^{\int_t^s c(X_r, r)dr} ds + \\
        \varphi(X_T) e^{\int_t^T c(X_r, r)dr} \Big].
    \end{multline}
    Furthermore, if Eq.s (\ref{eq:fk_pde}) and (\ref{eq:fk_boundary}) admit a classical solution, then \eqref{eq:fk_expectation} provides that classical solution.
\end{theorem}
For further details regarding Feynman-Kac, its derivation, and applications in control and beyond, we refer the interested reader to \cite{yong1997relations, kharroubi2015feynman, theodorou2010generalized} and works cited therein.

In our case, this leads to the following result:
\begin{theorem}\label{thm:desirability}
The optimal desirability score for Eq.s (\ref{eq:fk_pde}) and (\ref{eq:fk_boundary}) is given by 
    \begin{equation}\label{eq:desirability_expectation}
        \Psi(x, t) = \mathbb{E}\left[\psi(X_T)e^{\int_t^T -\frac{1}{\lambda}\ell(X_r)dr}\right]
    \end{equation}
    where $X_t$ is the solution of the SDE
    \begin{equation}\label{eq:desirability_sde}
    \begin{gathered}
        dX_s = f(X_s)ds + G(X_s)\sqrt{\lambda}LdW_s, \\
        X_0 = x,
    \end{gathered}
    \end{equation}
    where $L$ is the Cholesky decomposition of $R^{-1}$ (e.g., $R^{-1} = LL^T$) and $W_s$ is standard Brownian noise.
\end{theorem}
\begin{proof}
    This follows directly from Theorem~\ref{thm:fk}, and the fact that (\ref{eq:desirability_pde},~\ref{eq:desirability_boundry}) is a special case of (\ref{eq:fk_pde},~\ref{eq:fk_boundary}) with
    \begin{align*}
        \nu(x, t) &= \Psi(x, t), \quad
        c(x, t) = - \frac{1}{\lambda}\ell(x), \\
        b(x, t) &= f(x), \quad
        h(x, t) = 0, \quad
        \varphi(x) = \psi(x), \\
        \sigma(x)\sigma(x)^T &= \lambda GR^{-1}G^T
        = \left(\sqrt{\lambda} G L \right)\left(\sqrt{\lambda} G L \right)^T.
    \end{align*}
\end{proof}

Theorem~\ref{thm:desirability} offers a way to generate noisy samples from the desirability function \textit{without knowing the optimal policy}. All we need to do is sample from \eqref{eq:desirability_sde}, record the running and terminal costs, and use these recorded values to estimate \eqref{eq:desirability_expectation}.

Drawing samples from the stochastic process \eqref{eq:desirability_sde} is as simple as simulating the (deterministic) system,
\begin{equation}\label{eq:deterministic_dynamics}
    \dot{x} = f(x) + G(x)u,
\end{equation}
%
where the inputs $u$ are driven by a zero-mean Wiener process with covariance $\lambda R^{-1}$. Modern simulation software can generate such data extremely rapidly, even for large systems with complex dynamics. For example, MuJoCo can simulate a 27-degree-of-freedom humanoid robot at over 4000 steps per second on a CPU \cite{howell2022predictive}. With the addition of parallel computation on a GPU or TPU, this number increases to several million steps per second \cite{mjx}.

%% file: 05_supervised.tex
\section{A Supervised Learning Algorithm}\label{sec:supervised}

In this section, we present a simple supervised learning algorithm based on Theorem~\ref{thm:desirability}. For simplicity, we focus on learning the stationary initial policy
\begin{equation}
    \pi(x) = - R^{-1} G(x)^T \nabla_x V(x, 0).
\end{equation}
Applying this policy is similar in spirit to model predictive control (MPC) over a receding horizon. 


Algorithm~\ref{alg:supervised} summarizes the procedure. We begin by sampling $N$ states $x^i$ densely over $\Omega$. These states will define training data points, and could be sampled randomly or on a grid. For each state, we perform $M$ simulations of \eqref{eq:deterministic_dynamics} with random control inputs. Each of these simulations corresponds to an instance of the stochastic process \eqref{eq:desirability_sde}. For each simulation, we record the running cost $\ell$ and terminal cost $\phi$. 

\begin{algorithm}
    \caption{Supervised Policy Optimization}\label{alg:supervised}
    \KwIn{System \eqref{eq:dyanmics}, cost functional \eqref{eq:cost}, $\theta$, $\lambda$}
    \KwOut{Policy $\pi_\theta : \Omega \to \mathbb{R}^m$}

    $x^i \sim \mathrm{Uniform}(\Omega), \quad i \in [1..N]$
    
    \For{$j \in [1..M]$}{
        Simulate $\forall i$:
        \begin{align*}
            & \dot{x}^{(i,j)}_t = f(x^{(i,j)}_t) + G(x^{(i,j)}_t)u_t, \quad t \in [0, T] \\
            & u_t \sim \mathcal{N}(0, \lambda R^{-1}), \quad x_0^{(i,j)} = x^j.
        \end{align*}
        Average: $\hat{\Psi}^i = \frac{1}{M}\sum_{j=1}^M\psi(x^{(i,j)}_T) e^{-\int_t^T \frac{1}{\lambda}\ell(x^{(i,j)}_r)dr}$
    }

    $\theta \gets \mathrm{argmin}_\theta \frac{1}{N}\sum_{i=1}^N\| \Psi_\theta(x^i) - \hat{\Psi}^i \|^2$
    \vspace{0.3em}

    $\pi_\theta(x) \gets \lambda R^{-1}G(x)^T\frac{\nabla_x\Psi_{\theta}(x)}{\Psi_{\theta}(x)}$
    \vspace{0.1em}

    \KwRet{$\pi$}
\end{algorithm}

With this simulation data, we compute a Monte-Carlo estimate of \eqref{eq:desirability_expectation}, $\hat{\Psi}^i$, at each sample point $x^i$. We then fit a parameterized model $\Psi_\theta(x)$ to this training data by minimizing
\begin{equation}\label{eq:loss_function}
    \mathcal{L}(\theta) = \frac{1}{N}\sum_{i=1}^N\| \Psi_\theta(x^i) - \hat{\Psi}^i \|^2
\end{equation}
with stochastic gradient descent (SGD) \cite{bottou2018optimization}. Any smooth and sufficiently expressive function approximator could be used for the model $\Psi_\theta(x)$. In our examples, we choose $\Psi_\theta(x)$ to be a neural network with weights and biases $\theta$, but other methods (Gaussian process, kernel regression, etc.) could also be used. 

The optimal policy is then approximated as,
\begin{equation}\label{eq:policy}
    \pi(x) \approx \pi_\theta(x) = \lambda R^{-1} G(x)^T \frac{\nabla_x \Psi_\theta(x)}{\Psi_\theta(x)},
\end{equation}
which follows directly from \eqref{eq:value_policy} and \eqref{eq:desirability_defn}.

Of course, the effectiveness of this policy is limited by the degree to which $\Psi_\theta$ is a faithful approximation of the true desirability score $\Psi$. Here there is no free lunch. Even with infinite data and a suitably expressive model, $\mathcal{L}(\theta)$ is non-convex, and we cannot guarantee convergence to $\Psi$.

Nonetheless, there is reason for hope. For one thing, SGD for neural networks has demonstrated impressive performance on complex regression problems in computer vision and natural language processing. These problems are also non-convex, but SGD tends to find good local minima.

Similarly, while we cannot guarantee global convergence to the optimal policy, we can certify local convergence under mild conditions. In particular, consider an SGD procedure \cite[Algorithm 4.1]{bottou2018optimization} where the parameters $\theta$ are updated as
\begin{equation}\label{eq:sgd}
    \theta_{k+1} \gets \theta_k - \alpha_k g(\theta_k, \xi_k)
\end{equation}
at each iteration $k$. Here $\xi_k$ is a random variable representing a seed for batch selection, $\alpha_k > 0$ is a step size, and
\begin{equation}
   \mathbb{E}_{\xi_k}[g(\theta_k, \xi_k)]  = \nabla \mathcal{L}(\theta_k).
\end{equation}

\begin{theorem}\label{thm:convergence}
    Suppose that the loss function \eqref{eq:loss_function} is minimized with SGD \eqref{eq:sgd}, and that the step sizes $\alpha_k$ satisfy
    \begin{equation}
        \sum_{k=1}^\infty \alpha_k = \infty \quad \mathrm{and} \quad \sum_{k=1}^\infty \alpha_k^2 < \infty.
    \end{equation}
    If $\Psi_\theta$ is continuously differentiable, $\nabla_\theta \Psi_\theta$ is Lipschitz continuous, and there exist scalars $M, M_V \geq 0$ such that
    \begin{multline}\label{eq:variance_bound}
        \mathbb{E}_{\xi_k}[\|g(\theta_k, \xi_k)\|^2] - \|\mathbb{E}_{\xi_k}[g(\theta_k, \xi_k)^2]\|^2 \\
        \leq M + M_V \|\nabla\mathcal{L}(\theta_k)\|^2,
    \end{multline}
    then 
    \begin{equation}\label{eq:liminf}
        \liminf_{k \to \infty} \mathbb{E}[\|\nabla\mathcal{L}(\theta_k)\|^2] = 0.
    \end{equation}
    Furthermore, if $\Psi_\theta$ is infinitely differentiable w.r.t. $\theta$, then
    \begin{equation}\label{eq:lim}
        \lim_{k \to \infty} \mathbb{E}[\|\nabla \mathcal{L}(\theta_k)\|^2] = 0.
    \end{equation}
\end{theorem}
\begin{proof}
    This follows from the general supervised learning convergence results presented in \cite[Theorem 4.9]{bottou2018optimization}. 
    
    For \eqref{eq:liminf} to hold, we must show the following: (1) the loss function $\mathcal{L}(\theta)$ is continuously differentiable with Lipschitz continuous derivatives, (2) the step direction is an unbiased estimate of the gradient descent direction, (3) the variance of this estimate is bounded, and (4) $\mathcal{L}(\theta)$ is bounded from below. 

    Condition (1) is satisfied as long as $\Psi_\theta$ is differentiable with Lipschitz derivatives. Condition (2) is trivially satisfied by SGD methods of the form \eqref{eq:sgd}. Condition (3) is more difficult to enforce a priori, and gives rise to the assumption \eqref{eq:variance_bound}. Finally, condition (4) is trivially satisfied by mean-squared-error loss functions of the form \eqref{eq:loss_function}.

    Under these same conditions, the stronger result \eqref{eq:lim} holds by \cite[Corollary 4.12]{bottou2018optimization} as long as $\nabla \mathcal{L}(\theta)$ is continuously differentiable and the mapping $\theta \mapsto \|\nabla \mathcal{L}(\theta)\|^2$ has Lipschitz-continuous derivatives: this is trivially satisfied if $\Psi_\theta$ is infinitely differentiable with respect to $\theta$.
\end{proof}

Intuitively, Theorem~\ref{thm:convergence} states that as long as the model $\Psi_\theta(x)$ is suitably smooth, Algorithm~\ref{alg:supervised} will converge to a local minimizer of \eqref{eq:loss_function}. See \cite{bottou2018optimization} for further technical detail.

This convergence theorem may seem somewhat academic: after all, a local minimizer of \eqref{eq:loss_function} does not necessarily give rise to a good policy. However, most learning-based optimal control methods do not admit even this weak guarantee. Neural fitted value iteration \cite{lutter2021value}, PPO \cite{schulman2017proximal}, deep deterministic policy gradients \cite{lillicrap2015continuous}, and other popular RL algorithms involve coupled data collection and optimization, which precludes these standard supervised learning guarantees. As a result, these techniques typically require considerable hyperparameter tuning to work well. Small changes in hyperparameters or even the random seed can cause the optimization process to diverge. In contrast, Algorithm~\ref{alg:supervised} shares the desirable numerical properties enjoyed by supervised learning in other domains.

%% file: 06_examples.tex
\section{Examples}\label{sec:examples}

This section demonstrates our approach with two examples.

\subsection{Double Integrator}

\begin{figure}
    \centering
    \includegraphics[width=0.7\linewidth]{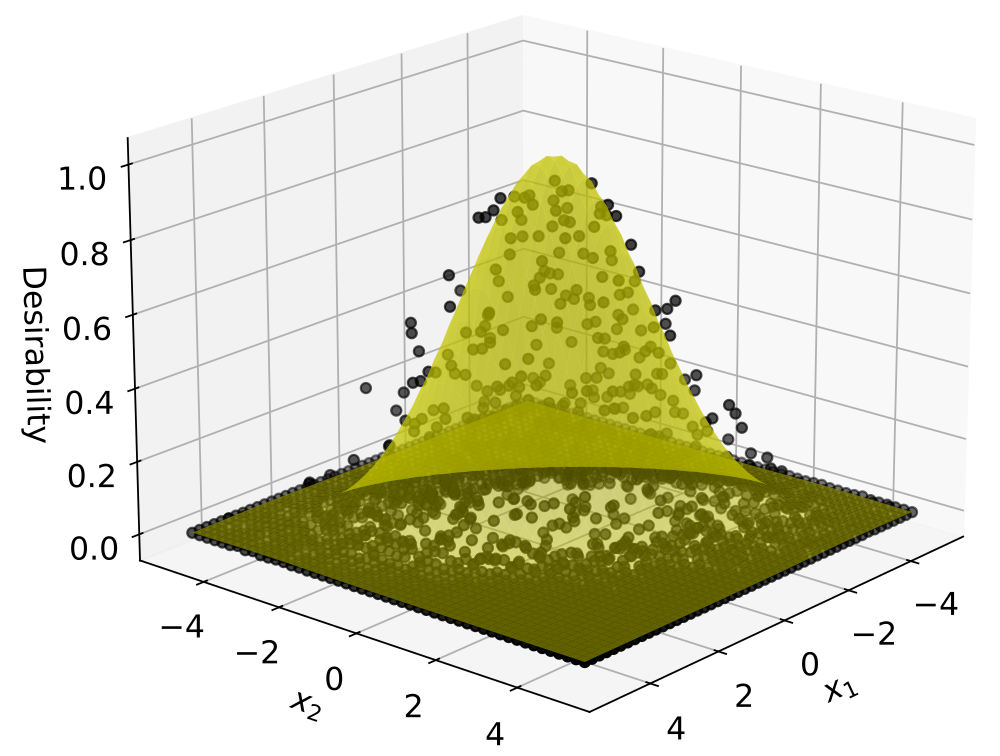}
    \caption{The true desirability function (yellow) and samples from the stochastic process \eqref{eq:desirability_sde} (black) for a double integrator. For systems where a closed-form solution is not available, these samples can serve as training data.}
    \label{fig:double_integrator}
\end{figure}

A simple double integrator has governing dynamics
\begin{equation*}
    dx_t = \left(\begin{bmatrix} 0 & 1 \\ 0 & 0 \end{bmatrix} x_t + \begin{bmatrix} 0 \\ 1 \end{bmatrix} u_t \right)dt + \begin{bmatrix} 0 \\ 1 \end{bmatrix} d\omega_t
\end{equation*}
where $x \in \mathbb{R}^2$ represents the position and velocity of a particle, and $u$ is the acceleration. We focus on a simple quadratic cost
\begin{equation*}
    \ell(x) = \phi(x) = x^Tx, \quad R = \begin{bmatrix} 1 \end{bmatrix}.
\end{equation*}
The optimal policy, value function, and desirability score for this system are readily available via LQG.

In Fig.~\ref{fig:double_integrator}, we plot the true desirability $\Psi(x)$ together with samples generated by simulating the diffusion process \eqref{eq:desirability_sde}. Note that these samples provide a noisy estimate of $\Psi(x)$, since \eqref{eq:desirability_expectation} holds only in expectation. Nonetheless, the samples from this stochastic process match the elongated shape of $\Psi(x)$. 

\subsection{Pendulum}

Next, we apply our proposed approach to a nonlinear system where a closed-form optimal policy is not readily available: the inverted pendulum. The pendulum dynamics are given by
\begin{equation*}
    dx_t = \left(\begin{bmatrix} \dot{\theta} \\  \frac{g}{l} \sin(\theta) \end{bmatrix} + \begin{bmatrix} 0 \\ \frac{1}{m l^2} \end{bmatrix} u_t \right)dt + \begin{bmatrix} 0 \\ 1 \end{bmatrix} d\omega_t,
\end{equation*}
where $x = [\theta, \dot{\theta}]^T$ and $\theta$ is the angular deviation from the upright position. The pendulum has mass $m=1$ kg, length $l=1$ m, and $g=9.81$ $\mathrm{m/s^2}$. The control inputs $u$ are torques. 

We specify the goal of balancing at the upright $\theta = 0$ with
\begin{equation*}
    \ell(x) = \phi(x) = \theta^2 + \frac{1}{10} \dot{\theta}^2, \quad R = \begin{bmatrix}1\end{bmatrix}.
\end{equation*}
Even without torque limits, the optimal policy swings the pendulum to inject energy into the system, as shown in Fig.~\ref{fig:closed_loop}.

\begin{figure}
    \centering
    \includegraphics[width=0.75\linewidth]{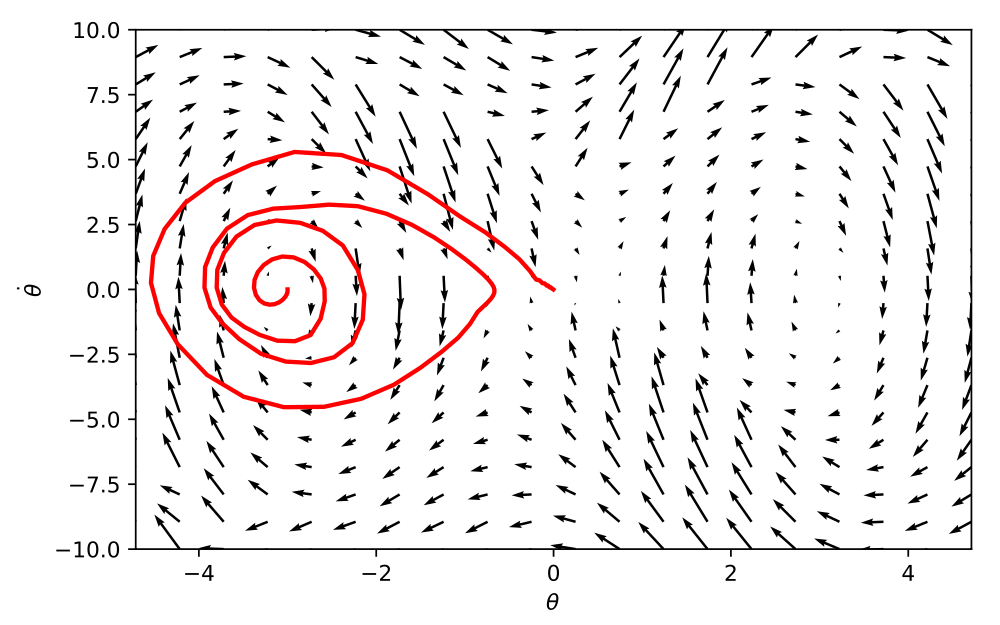}
    \caption{Closed loop vector field (black arrows) and simulated rollout (red) under the learned policy $\pi_\theta(x)$ for an inverted pendulum. This policy is generated in seconds using supervised learning, without access to demonstrations. }
    \label{fig:closed_loop}
\end{figure}

We apply Algorithm~\ref{alg:supervised} to learn an optimal policy, with $N = 10000$ data points, $M = 10$, and $\lambda = 20$. We simulate the diffusion process \eqref{eq:desirability_defn} with forward Euler integration and time step $\delta t = 0.01$ s for $T = 1.2$ seconds. Training data generated during this diffusion process is shown in Fig.~\ref{fig:diffusion_snaphots}. 

\begin{remark}
    $\lambda$ is the critical hyperparameter under this approach. If the value of $\lambda$ is chosen too small, $\Psi$ is dominated by a single uninformative peak. If chosen too large, the simulation rollouts become very noisy, requiring more samples and obscuring the structure of $\Psi$. 
\end{remark}

With this training data, we fit a neural network model of the desirability function $\Psi_\theta(x)$ in JAX. We used a simple fully-connected network with two hidden layers of 32 units and $\tanh$ activations. We applied 1000 epochs of Adam \cite{kingma2014adam} with an initial learning rate of $0.01$ and a batch size of 128. None of these hyperparameters are particularly optimized: this example serves merely as a proof-of-concept, and better results could almost certainly be obtained with a more careful application of supervised learning best practices. 

\begin{figure*}
    \centering
    \includegraphics[width=0.23\linewidth]{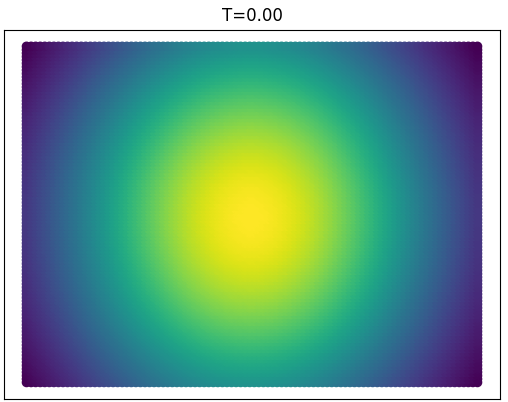}
    \includegraphics[width=0.23\linewidth]{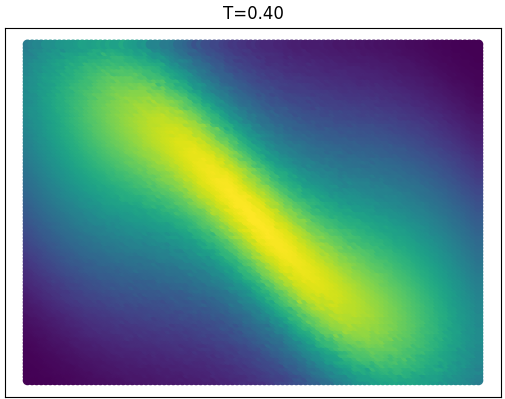}
    \includegraphics[width=0.23\linewidth]{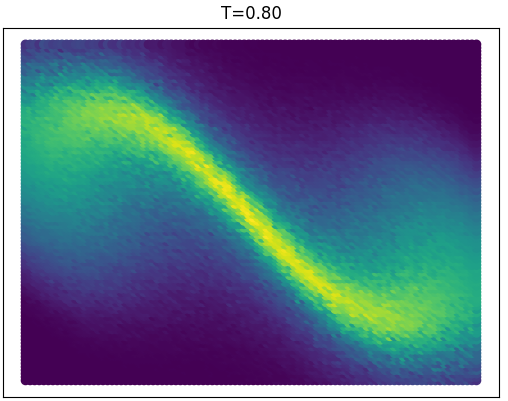}
    \includegraphics[width=0.23\linewidth]{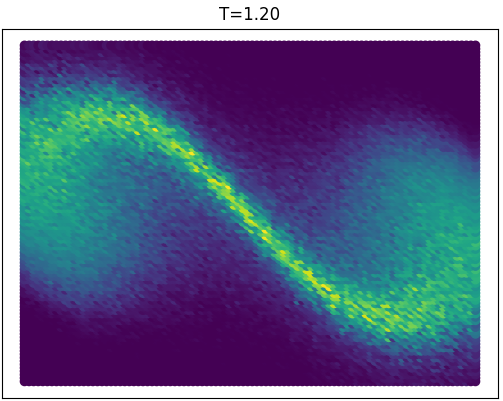}
    \caption{Snapshots of the noisy desirability (log value function) targets $\hat{\Psi}^i$ generated by the diffusion process \eqref{eq:desirability_sde}. These samples start out matching the terminal cost $\phi(x)$ at $T=0$ seconds. By $T=0.8$ seconds, the characteristic swirl of the value function is clearly visible.}
    \label{fig:diffusion_snaphots}
\end{figure*}

Altogether, the process of collecting training data and fitting the desirability function took a matter of seconds on a laptop. After this training process, applying the policy \eqref{eq:policy} in real time is trivial, as the desirability gradient $\nabla_x \Psi_\theta(x)$ can be efficiently computed by JAX's automatic differentiation. 

%% file: 08_conclusion.tex
\section{Conclusion and Future Work}\label{sec:conclusion}

In this paper, we showed how a class of nonlinear stochastic optimal control problems can be cast as supervised learning, without access to an optimal policy or demonstrations. In particular, we demonstrated how samples from the (log) optimal value function can be generated by simulating a stochastic process, via the Feynman-Kac theorem. Regression on these samples approximates the optimal value function and therefore the optimal policy. 
Future work will focus on scaling these techniques, expanding the class of nonlinear systems that can be considered, and developing software and numerical methods that scale to challenging optimal control problems, particularly in robot locomotion and contact-rich manipulation.